\documentclass{article}

\usepackage{fullpage}
\usepackage{amsmath, amsthm, amssymb}
\newtheorem{theorem}{Theorem}[section]
\newtheorem{lemma}[theorem]{Lemma}

\newtheorem{coro}[theorem]{Corollary}

\usepackage{amsmath, amssymb}

\usepackage{enumitem}
\setlist{itemsep=0pt, topsep=3pt}

\usepackage{algorithm}
\usepackage{algpseudocode}
\algtext*{EndWhile}
\algtext*{EndFor}
\algtext*{EndIf}

\usepackage{graphicx}
\graphicspath{{figs/}}
\usepackage{xcolor}
\usepackage{url}

\usepackage[colorlinks=true,citecolor=blue,linkcolor=purple]{hyperref}
\usepackage{cleveref}

\newcommand\set[1]{{\left\{#1\right\}}}

\newcommand{\poly}{{\mathrm{poly}}}

\newcommand{\calA}{{\mathcal{A}}}
\newcommand{\calB}{{\mathcal{B}}}

\newcommand{\calP}{{\mathcal{P}}}

\newcommand{\Real}{{\mathbb{R}}}

\DeclareMathOperator{\union}{\bigcup}

\title{Polynomial Integrality Gap of Flow LP for Directed Steiner Tree \footnote{A preliminary version of the paper appeared in the Proceedings of ACM-SIAM Symposium on Discrete Algorithms (SODA 2022). This version gives an explicit constant using approximations of binomial coefficients. This is better than the implicit one in the preliminary version obtained using Chernoff bounds.}}
\author{
	Shi Li\thanks{
		Department of Computer Science and Engineer, University at Buffalo, USA. \href{mailto:shil@buffalo.edu}{shil@buffalo.edu}.
		The work is supported in part by NSF grant CCF-1844890.
	}
	\and
	Bundit Laekhanukit\thanks{
		Institute for Theoretical Computer Science,
		Department of Information Management and Engineering,
		Shanghai University of Finance and Economics,  China. \href{mailto:lbundit+sufe@gmail.com}{lbundit+sufe@gmail.com}. This work is supported by Science and Technology Innovation 2030 –“New Generation of Artificial Intelligence” Major Project No.(2018AAA0100903), NSFC grant 61932002, Program for Innovative Research Team of Shanghai University of Finance and Economics (IRTSHUFE) and the Fundamental Research Funds for the Central Universities. 
		It is also supported by the national 1000-talent award by the Chinese government.
	}
}
\date{}

\begin{document}

	\maketitle
	
	%\fancyfoot[R]{\scriptsize{Copyright \textcopyright\ 2021 by SIAM\\ Unauthorized reproduction of this article is prohibited}}
	
	\begin{abstract}\small\baselineskip=9pt
		In the Directed Steiner Tree (DST) problem, we are given a directed graph $G=(V,E)$ on $n$ vertices with edge-costs $c \in \Real_{\geq 0}^E$, a root vertex $r \in V$, and a set $K \subseteq V \setminus \{r\}$ of $k$ terminals. The goal is to find a minimum-cost subgraph of $G$ that contains a path from $r$ to every terminal $t \in K$. DST has been a notorious problem for decades as there is a large gap between the best-known polynomial-time approximation ratio of $O(k^\epsilon)$ for any constant $\epsilon > 0$,  and the best quasi-polynomial-time approximation ratio of $O\left(\frac{\log^2 k}{\log \log k}\right)$.  
		
		Towards understanding this gap, we study the integrality gap of the standard flow LP relaxation for the problem. We show that the LP has an integrality gap of $\Omega(n^{0.0418})$. Previously, the integrality gap of the LP is only known to be $\Omega\left(\frac{\log^2n}{\log\log n}\right)$ [Halperin~et~al., SODA'03 \& SIAM J.~Comput.] and $\Omega(\sqrt{k})$ [Zosin-Khuller, SODA'02] in some instance with $\sqrt{k}=O\left(\frac{\log n}{\log \log n}\right)$. Our result gives the first known lower bound on the integrality gap of this standard LP that is polynomial in $n$, the number of vertices. Consequently, we rule out the possibility of developing a poly-logarithmic approximation algorithm for the problem based on the flow LP relaxation.
	\end{abstract}

	\section{Introduction}
	
	Network design problems play an important role in the area of combinatorial optimization both in theory and practice. 
	%The first and arguably most famous problem in the area is the {\em minimum spanning tree} problem (MST), which asks to find a minimum-cost subgraph spanning all the vertices. It has been a fundamental topic to study in the area of algorithm design. The analog of MST on directed graphs, also known as the {\em minimum-cost arborescence} problem also plays a huge part in illustrating the powerful concept of greedy algorithms in standard textbooks.
	%
	Two most famous problems in the area are the {\em Minimum Spanning Tree} problem, and its analogue in directed graphs, the {\em Minimum-Cost Arborescence} problem. They play a fundamental role in illustrating the powerful concept of greedy algorithms in standard textbooks.
	While the two textbook problems ask to find a tree or an arborescence that spans all the vertices, in most applications, it is more natural to connect only a subset of vertices, called {\em terminals}, which models clients in the network, while using some non-terminals or less-important nodes, called {\em Steiner vertices}, as relay. 
	This motivates the {\em Steiner tree} problems on both undirected and directed graphs, which have become the central focus in the area of network design for several decades. 
	Since the problems are known to be NP-hard, %and thus a clean greedy algorithm as in the case of minimum spanning tree is unlikely to exist.
	the algorithmic development has mainly been on finding good approximation algorithms. % rather than devising a fast algorithm as in the case of the minimum spanning tree problem. 

	In contrast to the Minimum Steiner Tree problem (on undirected graphs), the complexity status of the {\em Directed Steiner Tree} (DST) problem is much less-understood.  Formally, in the Directed Steiner Tree problem, we are given a directed graph $G = (V, E)$ on $n = |V|$ vertices, with edge costs $c \in \Real_{\geq 0}^E$, a root $r \in V$ and a set $K \subseteq V\setminus \{r\}$ of $k = |K|$ terminals. The goal of the problem is to find a minimum cost subgraph $T$ of $G$ that contains a path from $r$ to $t$, for every terminal $t \in K$. By minimality, we may assume WLOG that the solution subgraph $T$ is a tree (more precisely, an out-arborescence) rooted at $r$.

	 The best known polynomial-time approximation ratio for the problem is only $O(k^\epsilon)$ with a running time of $n^{O(1/\epsilon)}$ for any constant $\epsilon > 0$ \cite{Zelikovsky97, CharikarCCDGGL99}. %, where $n$ is the number of vertices in the graph, and $k$ is the number of terminals we need to connect. 
	 With quasi-polynomial time algorithms, one can achieve an approximation ratio of $O\left(\frac{\log ^2k}{\log \log k}\right)$ \cite{GrandoniLL-STOC19,GhugeN20}. It is known from the work of Halperin and Krauthgamer \cite{HalperinK03} that unless $\text{NP} \subseteq \text{ZPTIME}(n^{\poly\log(n)})$, there is no polynomial time  $O(\log^{2-\epsilon}k)$-approximation for the problem for any constant $\epsilon > 0$.  The long-standing open question on DST is whether a poly-logarithmic approximation for the problem can be achieved in polynomial time.

%	The polylogarithmic approximation algorithms for the problem is known only for those that run in quasi-polynomial running-time, while polynomial-factor is the best-known for polynomial-time algorithms. 
	%It has been a long standing open problem whether sub-polynomial approximation ratio is possible to obtained within polynomial running-time. 
	%Nevertheless, to the best of our knowledge, all the attempts to design a polynomial-time sub-polynomial approximation algorithms for DST have been failed. 
	%
	%
	In this paper, we study the integrality gap of the standard flow-based relaxation of DST, in the hope to shed some light on the approximability of DST under polynomial-time algorithms. 
	%
	%To be more precise, in the {\em directed Steiner tree} problem, we are given a directed graph $G=(V,E)$ on $n$ vertices with costs on edges, a root vertex $r$, and a set of $k$ terminals $T\subseteq V$, and the goal is to find a minimum-cost subgraph $H\subseteq G$ that has a directed path from the root vertex $r$ to every terminal $T$ (which we can assert such a solution subgraph to be arborescence).
	%
	%It is known that DST admits a polynomial-time $k^{\epsilon}$-approximation algorithm, for every constant $\epsilon>0$ \cite{CharikarCCDGGL99} using either recursive greedy algorithms or LP-based approximation algorithms through LP (resp, SDP) hierarchies\cite{Rothvoss11,FriggstadKKLST14}.
	%Although LP-based $k^{\epsilon}$-approximation algorithms are known, the standard LP relaxation of the problem has an integrality gap of at least $\sqrt{k}$ , and the LP-hierarchies technique incurred a running time of $n^{1/\epsilon}$. 
	The interesting part on the flow LP is that while its integrality gap is known to be lower bounded by $\Omega(\sqrt{k})$ \cite{ZosinK02}, the parameter $k$ (i.e., the number of terminals) in the construction is only $O\left(\frac{\log^2 n}{\log^2\log n}\right)$. Thus, the lower bound given by the instance of \cite{ZosinK02} is only $\Omega\left(\frac{\log n}{\log \log n}\right)$. Indeed,  the hardness result of \cite{HalperinKKSW07} implies an  integrality gap lower bound of  $\Omega\left(\frac{\log^2 n}{\log\log n}\right)$, which is better than that of \cite{ZosinK02} when we concern the dependence on $n$.
	
	%In our study, we attempt to prove the integrality ratio in terms of $n$, the number of vertices in the graph. At first, we believe in the integrality of polylogarithmic on $n$, which would give a polylogarithmic {\em estimation algorithm} for DST. 
	% Nevertheless, the answer to this question turns out negatively, and we discover a construction with polynomial lower bound. The construction itself also resemblances to that in the approximation hardness of the {\em Densest $k$-Subgraph} (DkS) problem, which has a similar property: There exists a polylogarithmic approximation algorithm that runs in quasi-polynomial-time, but there is no sub-polynomial approximation algorithm with polynomial running time unless the {\em Gap Exponential-Time Hypothesis} (Gap-ETH) is false \cite{Manurangsi17}.

\subsection{Our Results}
In this paper, we show that the integrality gap of the flow LP is lower bounded by $\Omega(n^{0.0418})$, a polynomial function of $n$.
%Here we state our result and contributions. Firstly, we show that an integrality gap of the standard flow-based LP relaxation for DST is polynomial on $n$. 
Our construction resemblances the previous one by Zosin and Khuller \cite{ZosinK02}, in the sense that our instance also has a 5-level structure.  In the instance of \cite{ZosinK02},  each vertex in the graph corresponds to some subset of $K$ (the set of terminals) of size roughly $\sqrt{k}$. This leads to an exponential dependence of $n$ on $k$.  In our construction, both the terminals and Steiner vertices correspond to subsets of some ground set, so $n$ and $k$ are polynomially related.  Crucially, we can still prove the useful properties that are needed to show the integrality gap.  The formal statement of our result is stated in Theorem~\ref{thm:main}, after we formally defined the flow LP relaxation \eqref{FLP}. As a consequence of our result, we rule out poly-logarithmic approximation algorithms for DST based on the flow LP.

 %The improvement on the integrality ratio can be cast as a gap amplification under ``parallel repetition''. This is a rare construction as most lower bound strengthening technique in terms of integrality ratio generally uses a recursive composition, which can be cast as a ``sequential repetition''. 
%As we go through an uncommon path, our analysis is quite involved, and we have to deal with a threshold argument. 

%Our result can be stated as follows.

%\begin{theorem} \label{thm:main}
%	Consider the flow-based LP relaxation of the directed Steiner tree problem as in \Cref{fig:flow-lp}. There exists an infinite family of instances whose integrality ratio is at least XXXX.
%\end{theorem}
%

%In other words, it is not possible to exploit the different on $k$ and $n$ to get a substantially better approximation. 
%
%In addition, our construction suggests that DST may have a behavior similar to DkS, although at the current state of the work, we cannot draw the explicit connections. Shortly, it suggests that a polynomial-time sub-polynomial approximation for DST is unlikely to exists.

\subsection{Related Work}
\label{sec:related-work}
	Compared to DST, the Minimum Steiner Tree problem (in undirected graphs) is better-understood.  Simply computing the minimum spanning tree over the metric closure of the terminals leads to a 2-approximation algorithm for the problem. %via MST heuristic has became a textbook standard, 
	The current best approximation ratio for the problem is $\ln(4) + \epsilon < 1.39$ due to Byrka et~al.\ \cite{BGR10}. On the negative side, the problem is known to be APX-hard \cite{BP89, CC08}. 

The DST problem has been a subject of studies for decades. The first non-trivial approximation result on this problem is due to Zelikovsky \cite{Zelikovsky97}, which gives an $O(k^\epsilon)$-approximation in $n^{O(1/\epsilon)}$-time for any constant $\epsilon > 0$, which is designed for directed acyclic graphs. The result is later extended in \cite{CharikarCCDGGL99} to general graphs. % thus giving an $O\left(\frac{k^{\epsilon} \log k}{\epsilon^2}\right)$ approximation algorithm that runs in $O(n^{1/\epsilon})$ time. 
Moreover, the $\epsilon$ in \cite{CharikarCCDGGL99} can be set to $1/\log k$, leading to an $O(\log^3k)$-approximation algorithm for DST in $n^{O(\log k)}$ time; this is the first poly-logarithmic approximation algorithm for the problem in quasi-polynomial time. A similar result was obtained by Kortsarz and Peleg \cite{KortsarzP99}. The approximation ratio in the quasi-polynomial-time regime has been improved to $O\left(\frac{\log^2k}{\log \log k}\right)$ in \cite{GrandoniLL-STOC19,GhugeN20}. Grandoni et~al.\ \cite{GrandoniLL-STOC19} also showed that this is the best approximation guarantee for any quasi-polynomial-time algorithm unless $\mathrm{NP}\subseteq\bigcup_{c>0}\mathrm{BPTIME}(2^{n^c})$ or the {\em Projection Game Conjecture} \cite{Moshkovitz15} is false.

There have been quite a few studies devoted to understand the power of LP/SDP hierarchy for the problem. %It was shown in \cite{ZosinK02} that the integrality gap of DST is at least $\Omega(\sqrt{k})$ even on a {\em 5-layered graph}. However in the instance that gives the lower bound, we have $k = O(\sqrt{\log n})$. Thus the lower bound is not polynomial in $n$, the size of the instance. 
It was shown in \cite{Rothvoss11} and \cite{FriggstadKKLST14} that the integrality gap of some basic LP relaxations can be brought down to $O(L\log^2 k)$ for DST on {\em $L$-level} graphs, if we lift them by $O(L)$ levels using some well-known LP/SDP hierarchies such as Sherali-Adams,  Lov\'asz–Schrijver and Lasserre hierarchies.  The quasi-polynomial time $O\left(\frac{\log^2k}{\log \log k}\right)$-approximation of \cite{GrandoniLL-STOC19} is also based on the Sherali-Adams hierarchy. 
%
%The undirected counter part of DST (abbreviatedly, ST) is also well-studied and stays in among the top profile problems.. The first known (folklore) approximation algorithm for this problem gives an approximation ratio, which has been a huge barrier at the time until the breakthrough result of Zelikovsky in \cite{Zelikovsky93} followed by subsequent works, e.g., \cite{RobinsZ05,XXX}. To date, the best approximation ratio for the undirected Steiner tree problem is XXX due to \cite{XXX}, which is an LP-based algorithm. Most of the modern works on ST are dedicated on the LP-based algorithms. Many LP formulation have been formulated to tighten the LP, and some of which have been later shown to be equivalent XXX.

%\paragraph*{Organizations.} 
%Prelim in XXX. Zosin-Khuller in XXX. New Construction in XXX. More Discussions in XXX.

\section{Flow LP Relaxation for Directed Steiner Tree}
\label{sec:prelim}

The standard flow LP relaxation is given in \eqref{FLP}. In the correspondent integer program, for every $e \in E$, we have a variable $x_e$ indicating whether $e$ is in the output tree $T$.  For every terminal $t \in K$, we let $\calP_t$ be the set of simple paths in $G$ from $r$ to $t$. Then for every $t \in K, P \in \calP_t$, we have a variable $f^t_P$ indicating if the $r\rightarrow t$ path in $T$ is $P$.  In the LP, we relax the 0/1-constraints on variables to non-negativity constraints \eqref{LPC:non-negative} and \eqref{LPC:non-negative-1}.

\noindent
\begin{minipage}[t]{\columnwidth}
	\begin{align}
		\min && \sum_{e\in E}c_ex_e \tag{\text{FLP}} \label{FLP}\\
		\text{s.t.} \nonumber\\
			 && \sum_{P\in\calP_t}f^t_P &= 1
		       & \forall t\in K \label{LPC:connected}\\
		     && \sum_{P\in\calP_t: e\in P}f^t_P &\leq x_e
		       & \forall e\in E, t\in K\label{LPC:capacity}\\
			 && x_e &\geq 0 
			   & \forall e \in E \label{LPC:non-negative}\\
			&& f^t_P & \geq 0 
			   &\forall t \in K, P \in \calP_t \label{LPC:non-negative-1}
	\end{align}
\end{minipage}
\bigskip

\Cref{LPC:connected} says that exactly one $r\to t$ path should exist in $T$, for every $t \in K$.
\Cref{LPC:capacity} says that if the $r\to t$ path $P$ in the output tree $T$ uses an edge $e$, then $e$ must be in $T$.  So, the constraints are valid.    In the LP, we can view $f^t_P$ as the amount of flow sending from $r$ to $t$ using a path $P$. Then \Cref{LPC:connected} and \Cref{LPC:capacity} together are equivalent to saying that the maximum flow from $r$ to $t$ in the graph $G$ with capacities $(x_e)_{e \in E}$ has value at least $1$, for every terminal $t \in K$. For this reason, we call \eqref{FLP} the flow LP. It can be solved efficiently despite the fact that it has exponential number of variables since the requirement for every $t$ is just a maximum-flow problem, which is equivalent to a polynomial-sized LP. 

Zosin and Khuller \cite{ZosinK02} showed that \eqref{FLP} has an integrality gap of $\Omega(\sqrt{k})$, for a family of instances with $\sqrt{k} = O\left(\frac{\log n}{\log \log n}\right)$. Therefore, the dependence of the integrality gap on $k$ is only logarithmic. In this paper, we show the following theorem:
\begin{theorem}
	\label{thm:main}
	The integrality gap of \eqref{FLP} is $\Omega(n^{0.0418})$.
\end{theorem}

\paragraph*{Organization.} The remainder of this paper is organized as follows. We first introduce the 5-level instance of Zosin and Khuller \cite{ZosinK02} in Section~\ref{sec:ZK}, where we give some necessary properties to prove an integrality gap. We give our construction satisfying the properties in Section~\ref{sec:proof}, which finishes the proof of Theorem~\ref{thm:main}.

	\section{Zosin-Khuller Type Gap Instance} \label{sec:ZK}
		In this section, we define the Zosin-Khuller type gap instance for the flow LP of the directed Steiner tree problem. To define such an instance, we need to specify the following objects. 
		\begin{enumerate}[label=(P\arabic*)]
			\item \label{objects-start}  $H = (\calA \uplus \calB, E_H)$ is a bipartite graph with $|\calA| \leq |\calB|$,  where all vertices in $\calA$ have the same degree $d$ and all vertices in $\calB$ have the same degree $d'$. 
			\item $K$ is a set of $k$ terminals, and $\{M_t\}_{t \in K}$  is a partition of $E_H$ into $k$ matchings of size $s$ each, one for each terminal $t \in K$. Therefore, we have $s\leq |\calA| \leq |\calB|, k \geq d \geq d'$ and $sk = d|\calA| = d'|\calB| = |E_H|$.  For convenience, we also treat each terminal $t \in K$ as a color and say all the edges in $M_t$ have color $t$.  So, the partition $\{M_t\}_{t \in K}$ gives a $k$-coloring of $E_H$. 
			\item \label{objects-end}  For every vertex $v \in \calB$, $K_v := \set{t \in K: v \text{ is matched in } M_t}$ is the set of colors of the incident edges of $v$. Notice that, for every $v \in \calB$, we have $|K_v| = d'$, and that $K_v$'s are determined by (P1) and (P2).
		\end{enumerate}
		
		Given the above objects, the Zosin-Khuller type gap instance is defined over a $5$-level directed graph $G = (V = V_0 \uplus V_1 \uplus V_2 \uplus V_3 \uplus V_4, E = E_1 \uplus E_2 \uplus E_3 \uplus E_4)$, where edges in $E_i$, $i \in \{1, 2, 3, 4\}$ go from $V_{i-1}$ to $V_i$.  The vertices and edges in $G$ are defined as follows. (See Figure~\ref{fig:layer} for an illustration of the instance.)
		\begin{itemize}
			\item $V_0 = \{r\}, V_1 = \calA, V_2 = \calB, V_3 = \calB'$ and $V_4 = K$, where $\calB'$ is a new set with the same cardinality as $\calB$. Let $\pi:\calB \to \calB'$ be a bijection from $\calB$ to $\calB'$. We say that the vertex $\pi(v) \in \calB'$ is the copy of the vertex $v \in \calB$, and thus, $\calB'$ is the copy of $\calB$.
			\item $E_1 = \{(r, u): u \in \calA\}$, $E_2  = E_H$ (except that edges in $E_2$ are directed), $E_3 = \set{(v, \pi(v)): v \in \calB}$, and $E_4 = \{(\pi(v), t): v \in \calB, t \in K_{v}\}$.
		\end{itemize}
		So, in the graph $G$, we have edges from $r$ to all vertices in $V_1 = \calA$, the graph $(V_1 \cup V_2, E_2)$ is just the graph $H$, and $(V_2 \cup V_3, E_3)$ is the matching between $\calB$ and its  copies. The only non-straightforward element is $E_4$: If the vertex $v \in \calB$ is incident to an edge with color $t \in K$ in $H$, then we have an edge from $\pi(v)$ to $t$ in $E_4$. Therefore, the set of $d'$ colors incident to a vertex $v \in \calB$ in $H$ is the same as the set of $d'$ out-neighbors of $\pi(v)$ in $G$. The in-degree of a terminal $t \in K$ is $s$. 
		
		The root of the DST instance is $r$, and the terminals are $V_4 = K$. To complete the DST instance, it remains to define the costs of edges in $G$: Edges in $E_1, E_2, E_3$ and $E_4$ have costs $\frac{|\calB|}{|\calA|}, 0, 1$ and $0$ respectively.  \medskip
		
	\begin{figure*}
		\centering
		\includegraphics[width=0.7\textwidth]{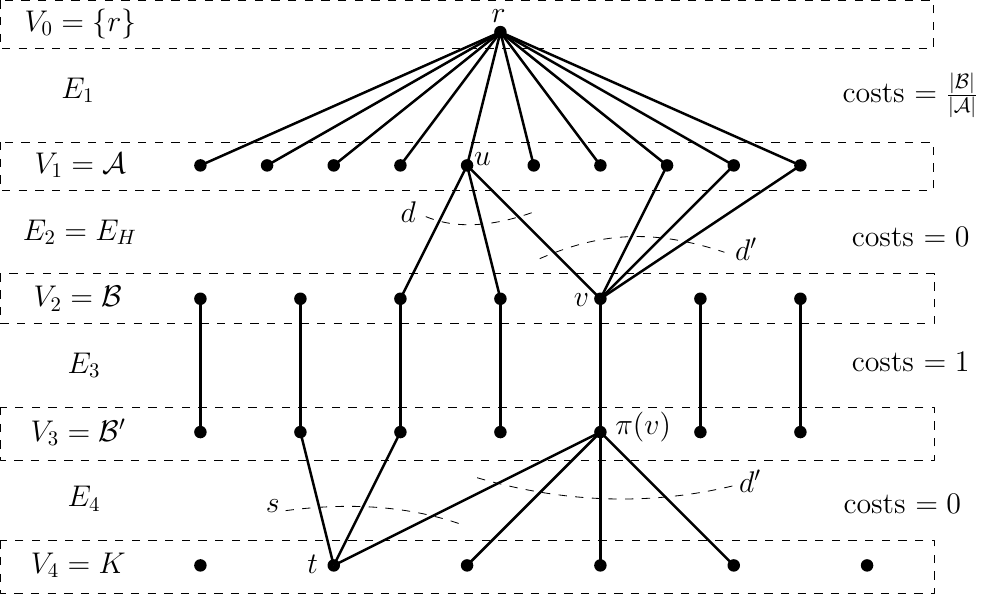}
		\caption{An illustration of the Zosin-Khuller type gap instance. All the edges go downwards. We only show some representative edges in $E_2$ and $E_4$. The set of colors of incoming edges of $v$ is precisely the set out-neighbors of $\pi(v)$.}
		\label{fig:layer}
	\end{figure*}
		
		Given the Zosin-Khuller type instance, we can naturally define an LP solution $x \in [0, 1]^E$, where every edge $e \in E$ has $x_e = \frac1s$.  The cost of the LP solution is $\frac1s\cdot\left(|E_1|\cdot\frac{|\calB|}{|\calA|} + |E_3|\cdot 1\right) = \frac1s\cdot\left(|\calA|\cdot\frac{|\calB|}{|\calA|} + |\calB|\cdot 1\right) = \frac{2|\calB|}{s}$. For every terminal $t \in K$, we can find $s$ disjoint paths from $r$ to $t$ in $G$: For each $(u, v) \in M_t$, we take the path $r\to u \to v \to \pi(v) \to t$.  Notice that the edge $(\pi(v), t)$ exists since $t \in K_v$. Therefore, the $x$ induces a valid solution to \eqref{FLP}. 
		
		The following lemma says that if the objects specified in \ref{objects-start} to \ref{objects-end} have a good property, then the DST instance has a large integrality gap. 
		\begin{lemma}
			\label{lemma:gap-property}
			Let $\alpha \geq 1$ be a real number. Suppose for every $u \in \calA$, there exists some $J_u \subseteq K$ of terminals with $|J_u| \leq \frac{d}{\alpha}$, such that for every $(u, v) \in E_H$, we have $|K_v \setminus J_u| \leq \frac{d'}{\alpha}$. Then the optimum solution to the DST instance has cost at least $\frac{\alpha|\calB|}{s}$.
		\end{lemma}
		\begin{proof}
			We can break the optimum directed Steiner tree $T^*$ of $G$ into many sub-trees, each containing exactly one edge in $E_1$. We show that any such sub-tree $T'$ has a small \emph{density}, which is defined as the number of terminals in $T'$ divided by the cost of the edges in $T'$.   In particular, we show its density is at most $\frac{d'}{\alpha}$.   If this holds, then the optimum tree has cost at least $k/\frac{d'}\alpha = \alpha \cdot \frac{k}{d'} = \alpha \cdot \frac{|\calB|}{s}$ as $ks = d'|\calB| = |E_H|$. 
			
			So, we fix a sub-tree $T'$ that contains exactly one edge in $E_1$, and it remains to show that the density of $T'$ is at most $\frac{d'}{\alpha}$. Let $u$ be the unique vertex in $V_1 = \calA$ in the tree.  Let $V'$ be the set of vertices in $V_2$ in the tree; without loss of generality we assume that, for every $v \in V'$, we have $(v, \pi(v)) \in T'$ since, otherwise, we can remove $(u, v)$ from the solution. So, each vertex in $V'$ is a neighbor of $u$ in $H$.  The set of terminals that can be reached from $V'$ is $\union_{v \in V'}K_v$. Therefore, the number of terminals in the tree $T'$ is at most 
			\begin{align*}
				\left|\union_{v \in V'}K_v\right|  = \left|J_u \cup \union_{v \in V'}(K_v \setminus J_u)\right| \leq \frac{d}{\alpha} + \frac{d'}{\alpha}|V'|.
			\end{align*}
			
			The cost of $T'$ is exactly $\frac{|\calB|}{|\calA|} + |V'| = \frac{d}{d'} + |V'|$, and the density of the tree is at most 
				$\displaystyle \frac{\frac{d}{\alpha} + \frac{d'}{\alpha}|V'|}{\frac{d}{d'} + |V'|} = \frac{d'}{\alpha}$.
		\end{proof}

		Therefore, if the condition in Lemma~\ref{lemma:gap-property} holds, then the DST instance has integrality gap at least $\alpha/2$. 
		
		In the Zosin-Khuller gap instance in \cite{ZosinK02}, the objects from \ref{objects-start} to \ref{objects-end} are defined as follows. Let $K$ be the set of $k$ terminals, and assume $\sqrt{k}$ is an integer. We have $\calA = {K \choose \sqrt{k}}$, $\calB = {K\choose \sqrt{k} + 1}$, and there is an edge from $A \in \calA$ to $B \in \calB$ in $H$ if and only if $A \subseteq B$. The color of $(A, B)$ is the unique terminal in $B \setminus A$.  Notice that all the required properties are satisfied. $d = k -\sqrt{k}, d' = \sqrt{k}+1$, and for a vertex $B \in \calB$, we have $K_B = B$.
		
		To satisfy the condition in Lemma~\ref{lemma:gap-property}, we define $J_A =  A$ for every $A \in \calA$. So, $|J_A| = \sqrt{k}$. Then, for every $B \in \calB$ that is adjacent to $A$, we have $K_B \setminus J_A = B \setminus  A$, which has size 1. Therefore, we can define $\alpha = \min\{\frac{d}{\sqrt{k}}, \frac{d'}{1}\} = \min\{\sqrt{k}-1, \sqrt{k}+1\} = \sqrt{k}-1$ to make the condition of Lemma~\ref{lemma:gap-property} holds.  The instance gives a gap of $\Omega(\sqrt{k})$. However, it has size exponential in $\sqrt{k}$. 

	\section{Instance with Polynomial Integrality Gap} 
	\label{sec:proof}
	The construction of our gap instance is similar to that of \cite{ZosinK02} in the sense that our $\calA$ and $\calB$ correspond to subsets of a ground set.  However, in our construction, a terminal also corresponds to a subset of the ground set (as opposed to a single element). In the graph $H$, we have an edge from some element $A \in \calA$ to some element $B \in \calB$ if and only if $A \subseteq B$, and the color of the edge is the set $B \setminus A$, which will be an element in $K$.  Using this construction, we make the number of terminals exponential in the size of the ground set. This can lead to a polynomial integrality gap. We carefully design the sizes of the subsets so that the conditions in Lemma~\ref{lemma:gap-property} hold with a large $\alpha$. 
	
	We formally define the objects specified in \ref{objects-start} to \ref{objects-end}.   Let $\rho \in (0, \frac14)$ and $\theta \in \left(\rho^2, \frac\rho 2\right)$ be two absolute rational numbers whose values will be decided later.  Let $m > 0$ be an integer so that $\rho m$ and $\theta m$ are integers. $[m]$ will be the ground set. The objects are defined as follows:
	\begin{itemize}
		\item $\calA = K=  {[m] \choose \rho m}, \calB = {[m] \choose 2\rho m}$, there is an edge from some $A \in \calA$ to some $B \in \calB$ in $H$ if and only if $A \subseteq B$. The edge has color $B \setminus A \in K$. 
		\item Therefore, $d = {(1-\rho)m \choose \rho m}, d' = {2\rho m \choose \rho m}, k = {m \choose \rho m}$ and $K_B = \set{C \subseteq B: |C| = \rho m}$ for every $B \in \calB$.
	\end{itemize}
	It is easy to check that all the required properties are satisfied: $|\calA| \leq |\calB|$, the graph $H$ is bi-regular  and all matchings have the same size, as $H$ is highly symmetric. \medskip
	
	Now we shall define the set $J_A$ for any $A \in \calA$ to satisfy the condition of Lemma~\ref{lemma:gap-property}: 
	\begin{align*}
		J_A = \{C  \in K: |C \cap A| \geq \theta m\}.
	\end{align*}

	%We then show that this definition satisfies the condition of Lemma~\ref{lemma:gap-property} for some $\alpha = e^{\Omega(m)}$.  From now on, we fix any $A \in \calA = {[m] \choose \rho m}$. 
	
	From now on we use $\log(\cdot)$ to denote $\log_2(\cdot)$, and $h(a):= -a\log a -(1-a)\log(1-a)$ for every $a \in [0, 1]$ to denote the entropy of the Bernoulli distribution with mean $a$. 	
	\begin{lemma}
		\label{lemma:stirling}
		For every two integers $0 \leq q \leq p$, we have 
		\begin{align*}
			\log {p \choose q} = h\left(\frac{q}{p}\right) \cdot p \pm O(\log p).
		\end{align*}
	\end{lemma}
	\begin{proof}
		Using Stirling's approximation, we have $\log_2(p!) = p \log_2 p - p\log_2e \pm O(\log p)$ for integers $p \geq 1$. 
		So, 
		\begin{align*}
			\log {p \choose q} = \log  \frac{p!}{q!(p-q)!} &= p \log  p - q \log  q - (p-q)\log  (p-q) \pm O(\log p)\\
			&= p\left(\frac{q}{p} \log \frac {p}{q} + \frac{p-q}{p}\log \frac{p}{p-q}\right) \pm O(\log p)\\
			&= h\left(\frac {q}{p}\right) \cdot p \pm O(\log p). \qedhere
		\end{align*}
	\end{proof}

	\begin{coro}
		\label{coro:bound-quantities}
		The following equalities hold:
		\begin{align*}
			\log |\calB| &= h(2\rho) \cdot m \pm O(\log m).\\
			\log d &= h\left(\frac{\rho}{1-\rho}\right)\cdot (1-\rho) m \pm O(\log m).\\
			\log d' &= 2\rho m \pm O(\log m).
		\end{align*}
	\end{coro}

	\begin{proof}
		Notice that  $|\calB| = {m \choose 2\rho m}$,  $d = {(1-\rho)m \choose \rho m}$ and $d' = {2\rho m \choose \rho m}$. Then applying Lemma~\ref{lemma:stirling} and using that $h(0.5)=1$ proves the corollary. 
	\end{proof}

	\begin{lemma}
		\label{lemma:bi-monotone}
		Let $q \leq p_1 \leq p_2$ be positive integers. Consider the function
		${p_1 \choose a}{p_2\choose q-a}$ over integers $a \in [0, q]$.  Then the function is increasing over $a \in \big[0, \frac{qp_1}{p_1 + p_2}\big]$ and decreasing over $a \in \big[\frac{qp_1}{p_1+p_2}, q\big]$.
	\end{lemma}	
	\begin{proof}
		For notation convenience, we define $f(a) = {p_1 \choose a}{p_2\choose q-a}$ for every $a \in [0, q]$. Then, fix $a \in [0, q)$ and we have 
		\begin{align*}
			\frac{f(a+1)}{f(a)}  = \frac{p_1 - a}{a+1} \times \frac{q-a}{p_2+1-(q-a)}.
		\end{align*}
		Notice that the right-side is a decreasing function of $a$.   For the quantity to be 1,  we need $\displaystyle \frac{p_1-a}{a+1} = \frac{p_2+1-q+a}{q-a}$, which is equivalent to $\frac{p_1+1}{a+1} = \frac{p_2+1}{q-a}$, and $(p_1+1)(q-a) = (p_2+1)(a+1)$. So,  we need $a$ to be $\frac{q(p_1+1)-(p_2+1)}{p_1+p_2+2} = \frac{(q+1)(p_1+1)}{p_1+p_2+2} - 1$. This is strictly between $\frac{qp_1}{p_1+p_2}-1$ and $\frac{qp_1}{p_1+p_2}$.  To see the lower bound, notice that $\frac{p_1+1}{p_1+p_2+2} \geq \frac{p_1}{p_1+p_2}$ as $0 < p_1 \leq p_2$.  To see the upper bound, note that $\frac{q(p_1+1)-(p_2+1)}{p_1+p_2+2} < \frac{qp_1}{p_1+p_2+2} < \frac{qp_1}{p_1+p_2}$. 
		
		So, when $a \leq \frac{qp_1}{p_1+p_2}-1$, we have $f(a+1) > f(a)$. When $a \geq \frac{qp_1}{p_1+p_2}$, we have $f(a+1) < f(a)$. 
	\end{proof}
	
	\begin{lemma} 
		\label{lemma:bound-log-JA-KB}
		For any $A \in \calA$ and any $B \in \calB$ with $(A, B) \in E_H$, we have
		\begin{align*}
			\log |J_A| &\leq  \left(h\left(\frac{\theta}{\rho}\right)\cdot \rho  + h\left(\frac{\rho - \theta}{1-\rho}\right)\cdot (1-\rho)\right) m + O(\log m).\\
			\log |K_B \setminus J_A| &\leq h\left(\frac\theta\rho\right)\cdot 2\rho m + O(\log m).
		\end{align*}
	\end{lemma}
	\begin{proof}
		Notice that $\displaystyle |J_A| = |\{C  \in K: |C \cap A| \geq \theta m\}| = \sum_{a = \theta m}^{\rho m}{\rho m \choose a}{ (1-\rho)m \choose \rho m - a}$.
		We have $\theta m \geq \frac{\rho m\cdot \rho m}{m} = \rho^2 m$ by our assumption. By Lemma~\ref{lemma:bi-monotone}, the maximum of ${\rho m \choose a}{ (1-\rho)m \choose \rho m - a}$ over $a \in [\theta m, \rho m]$ is achieved when $a = \theta m$.  So, 
		\begin{align*}
			\log |J_A| &\leq \log \left({\rho m \choose \theta m}{ (1-\rho)m \choose \rho m - \theta m}\right) + O(\log m) \leq h\left(\frac\theta\rho\right)\cdot \rho m + h\left(\frac{\rho - \theta}{1-\rho}\right)\cdot (1-\rho) m + O(\log m),
		\end{align*}
		by Lemma~\ref{lemma:stirling}.
	
		Similarly,  we have $\displaystyle |K_B \setminus J_A| = |\{ C \in K: C \subseteq B, |C \cap A| < \theta m \}| = \sum_{a = 0}^{\theta m - 1}{\rho m \choose a}{\rho m\choose \rho m - a}$. We have $\theta m \leq \frac{\rho m \cdot \rho m}{2\rho m} = \frac{\rho m}{2}$ by our assumption.  Hence, the maximum of ${\rho m \choose a}{\rho m\choose \rho m - a}$ over $a \in [0, \theta m)$ is at most ${\rho m \choose \theta m}{\rho m\choose \rho m - \theta m}$ by Lemma~\ref{lemma:bi-monotone}. Therefore, 
		\begin{align*}
			\log |K_B\setminus J_A| &\leq \log \left({\rho m \choose \theta m}{ \rho m \choose \rho m - \theta m}\right) + O(\log m) \leq h\left(\frac\theta\rho\right)\cdot \rho m + h\left(\frac{\rho - \theta}{\rho}\right)\cdot \rho m + O(\log m)\\
			&= h\left(\frac\theta\rho\right)\cdot 2\rho m + O(\log m)
		\end{align*}
		by Lemma~\ref{lemma:stirling} and that $h\left(\frac{\theta}{\rho}\right) = h\left(\frac{\rho - \theta}{\rho}\right)$.
		This finishes the proof of the lemma. 
	\end{proof}

	Therefore, to satisfy the condition in Lemma~\ref{lemma:gap-property}, we can define $\alpha$ as follows. In the definition, $(A, B)$ is any edge in $E_H$. Due to the symmetry, $|J_A|$ and $|K_B \setminus J_A|$ do not depend on the edge $(A, B)$.
	\begin{align*}
		\log \alpha := \min \left\{\begin{matrix}\log \frac{d}{|J_A|}\\  \log \frac{d'}{|K_B \setminus J_A|}\end{matrix}\right\} \geq
		\left\{\begin{matrix}
			h\left(\frac{\rho}{1-\rho}\right)\cdot (1-\rho)  - h\left(\frac{\theta}{\rho}\right)\cdot \rho  - h\left(\frac{\rho - \theta}{1-\rho}\right)\cdot (1-\rho) \\
			\left(1 - h\left(\frac\theta\rho\right)\right)\cdot 2\rho 
		\end{matrix}\right\}\cdot m - O(\log m).
	\end{align*}
	The inequality is by Corollary~\ref{coro:bound-quantities} and Lemma~\ref{lemma:bound-log-JA-KB}.

	With the assistance of computer programs, we set $\rho$ and $\theta$ as follows: Let $\rho = 0.007454$ and $\theta = 0.00136256$, which satisfy $\rho^2 \leq \theta \leq \frac{\rho}{2}$ and $\rho \leq \frac14$. Notice that $\log n = \log |\calB| + O(1) \leq h(2\rho)m + O(\log m)$ by Corollary~\ref{coro:bound-quantities}.  Therefore, if we define $\delta^\circ$ as follows:
	\begin{align*}
		\delta^\circ := \frac{
			\min\left\{h\left(\frac{\rho}{1-\rho}\right)\cdot (1-\rho)  - h\left(\frac{\theta}{\rho}\right)\cdot \rho  - h\left(\frac{\rho - \theta}{1-\rho}\right)\cdot (1-\rho), \quad \left(1 - h\left(\frac\theta\rho\right)\right)\cdot 2\rho\right\}
		}{h(2\rho)} > 0.0418,
	\end{align*}
	then we have $\frac{\log \alpha}{\log n} \geq \delta^\circ - O\left(\frac{\log m}m\right) = \delta^\circ - O\left(\frac{\log \log n}{\log n}\right)$. This implies $\alpha \geq n^{\delta^\circ} / (\log n)^{O(1)}$. As $\delta^\circ > 0.0418$, we have $\alpha = \Omega(n^{0.0418})$. By Lemma~\ref{lemma:gap-property}, the integrality gap of the flow LP \eqref{FLP} for our instance is $\Omega(n^{0.0418})$; this finishes the proof of Theorem~\ref{thm:main}. 
	
	\paragraph{Remark} One may notice that we could use different sizes for subsets in $\calA$ and in $K$.  However, with the assistance of computer programs, we see that the best constant is obtained when the sizes are the same. So, we let both of them be $\rho m$ in our instance. 
	
\section{Future Directions}
\label{sec:conclusion}
Our result rules out a polynomial-time poly-logarithmic approximation for DST using the flow LP relaxation.  Our instance is on a 5-level graph, and thus it admits a poly-logarithmic approximation via recursive greedy \cite{CharikarCGG98} or LP hierarchies \cite{Rothvoss11, FriggstadKKLST14}.  One interesting future direction is to understand the power of  $O(1)$-level lift of the flow LP using some hierarchy.  That is, whether the lifted LP still has a polynomial integrality gap.  To show a lower bound, we may need to lift our gap instance for the basic flow LP in some manner. 

\paragraph*{Acknowledgment} 

We thank Jittat Fakcharoenphol for the useful discussions.
	
\bibliography{dst.bib}

\newcommand{\etalchar}[1]{$^{#1}$}
\begin{thebibliography}{HKK{\etalchar{+}}07}

\bibitem[BGRS10]{BGR10}
Jaroslaw Byrka, Fabrizio Grandoni, Thomas Rothvo\ss{}, and Laura Sanit\`{a}.
\newblock An improved lp-based approximation for steiner tree.
\newblock In {\em Proceedings of the Forty-Second ACM Symposium on Theory of
  Computing}, STOC '10, page 583–592, New York, NY, USA, 2010. Association
  for Computing Machinery.

\bibitem[BP89]{BP89}
M.~Bern and P.~Plassmann.
\newblock The steiner problem with edge lengths 1 and 2.
\newblock {\em Inf. Process. Lett.}, 32:171--176, 1989.

\bibitem[CC08]{CC08}
Miroslav Chleb{\'i}k and J.~Chleb{\'i}kov{\'a}.
\newblock The steiner tree problem on graphs: Inapproximability results.
\newblock {\em Theor. Comput. Sci.}, 406:207--214, 2008.

\bibitem[CCC{\etalchar{+}}99]{CharikarCCDGGL99}
Moses Charikar, Chandra Chekuri, To{-}Yat Cheung, Zuo Dai, Ashish Goel, Sudipto
  Guha, and Ming Li.
\newblock Approximation algorithms for directed steiner problems.
\newblock {\em J. Algorithms}, 33(1):73--91, 1999.

\bibitem[CCGG98]{CharikarCGG98}
Moses Charikar, Chandra Chekuri, Ashish Goel, and Sudipto Guha.
\newblock Rounding via trees: Deterministic approximation algorithms for group
  steiner trees and \emph{k}-median.
\newblock In {\em Proceedings of the Thirtieth Annual {ACM} Symposium on the
  Theory of Computing, Dallas, Texas, USA, May 23-26, 1998}, pages 114--123,
  1998.

\bibitem[FKK{\etalchar{+}}14]{FriggstadKKLST14}
Zachary Friggstad, Jochen K{\"{o}}nemann, Young Kun{-}Ko, Anand Louis, Mohammad
  Shadravan, and Madhur Tulsiani.
\newblock Linear programming hierarchies suffice for directed steiner tree.
\newblock In {\em Integer Programming and Combinatorial Optimization - 17th
  International Conference, {IPCO} 2014, Bonn, Germany, June 23-25, 2014.
  Proceedings}, pages 285--296, 2014.

\bibitem[GLL19]{GrandoniLL-STOC19}
Fabrizio Grandoni, Bundit Laekhanukit, and Shi Li.
\newblock \emph{O}(log\({}^{\mbox{2}}\) \emph{k} / log log
  \emph{k})-approximation algorithm for directed steiner tree: a tight
  quasi-polynomial-time algorithm.
\newblock In {\em Proceedings of the 51st Annual {ACM} {SIGACT} Symposium on
  Theory of Computing, {STOC} 2019, Phoenix, AZ, USA, June 23-26, 2019}, pages
  253--264, 2019.

\bibitem[GN20]{GhugeN20}
Rohan Ghuge and Viswanath Nagarajan.
\newblock Quasi-polynomial algorithms for submodular tree orienteering and
  other directed network design problems.
\newblock In {\em To appear in SODA'20.}, 2020.

\bibitem[HK03]{HalperinK03}
Eran Halperin and Robert Krauthgamer.
\newblock Polylogarithmic inapproximability.
\newblock In {\em Proceedings of the 35th Annual {ACM} Symposium on Theory of
  Computing, June 9-11, 2003, San Diego, CA, {USA}}, pages 585--594, 2003.

\bibitem[HKK{\etalchar{+}}07]{HalperinKKSW07}
Eran Halperin, Guy Kortsarz, Robert Krauthgamer, Aravind Srinivasan, and Nan
  Wang.
\newblock Integrality ratio for group steiner trees and directed steiner trees.
\newblock {\em {SIAM} J. Comput.}, 36(5):1494--1511, 2007.

\bibitem[KP99]{KortsarzP99}
Guy Kortsarz and David Peleg.
\newblock Approximating the weight of shallow steiner trees.
\newblock {\em Discrete Applied Mathematics}, 93(2-3):265--285, 1999.

\bibitem[Mos15]{Moshkovitz15}
Dana Moshkovitz.
\newblock The projection games conjecture and the np-hardness of ln
  n-approximating set-cover.
\newblock {\em Theory of Computing}, 11:221--235, 2015.

\bibitem[PS97]{PanconesiS97}
Alessandro Panconesi and Aravind Srinivasan.
\newblock Randomized distributed edge coloring via an extension of the
  chernoff-hoeffding bounds.
\newblock {\em {SIAM} J. Comput.}, 26(2):350--368, 1997.

\bibitem[Rot11]{Rothvoss11}
Thomas Rothvo{\ss}.
\newblock Directed steiner tree and the lasserre hierarchy.
\newblock {\em CoRR}, abs/1111.5473, 2011.

\bibitem[Zel97]{Zelikovsky97}
Alexander Zelikovsky.
\newblock A series of approximation algorithms for the acyclic directed steiner
  tree problem.
\newblock {\em Algorithmica}, 18(1):99--110, 1997.

\bibitem[ZK02]{ZosinK02}
Leonid Zosin and Samir Khuller.
\newblock On directed steiner trees.
\newblock In {\em Proceedings of the Thirteenth Annual {ACM-SIAM} Symposium on
  Discrete Algorithms, January 6-8, 2002, San Francisco, CA, {USA.}}, pages
  59--63, 2002.

\end{thebibliography}
\bibliographystyle{alpha}

%\newpage

\appendix

\end{document}